\documentclass[a4paper,11pt]{article}

\usepackage{newlfont}
\usepackage[small]{caption2}
\usepackage[below]{placeins}
\usepackage{flafter}
\usepackage{amsfonts}
\usepackage{amsmath}
\usepackage{amssymb}
\usepackage[american]{babel}
\usepackage{graphicx}
\usepackage{subfigure}
\usepackage{multirow}
\usepackage{url}
\usepackage[linesnumbered,ruled,vlined]{algorithm2e}
\usepackage{color}
\newcommand{\sumid}[2]{\begin{subarray}{c}
#1\\
#2
\end{subarray}}

\allowdisplaybreaks

\newcommand{\Udec}{\bar{U}_{dec}}
\newcommand{\UPt}{U_{\vec{P}_t}}
\newcommand{\SumPT}{\sum_{\vec{P}_T}Pr(\vec{P}_T)}
\newcommand{\Sumt}{\sum_{t=0}^{T}}
\newcommand{\Sumk}{\sum_{k=1}^{n-m}}

\newcommand{\ket}[1]{\big| #1 \big\rangle}
\newcommand{\bra}[1]{\big\langle #1 \big|}
\newcommand{\braket}[2]{\big\langle #1 \big| #2 \big\rangle}                 
\newcommand{\bracket}[3]{\big\langle #1 \big| #2 \big| #3 \big\rangle}       

\newtheorem{theorem}{Theorem}[section]
\newtheorem{lemma}[theorem]{Lemma}

\newtheorem{corollary}[theorem]{Corollary}
\newtheorem{definition}[theorem]{Definition}

\newenvironment{proof}[1][Proof]{\begin{trivlist}
\item[\hskip \labelsep {\bfseries #1}]}{\end{trivlist}}

\newcommand{\qed}{\nobreak \ifvmode \relax \else
      \ifdim\lastskip<1.5em \hskip-\lastskip
      \hskip1.5em plus0em minus0.5em \fi \nobreak
      \vrule height0.75em width0.5em depth0.25em\fi}

\begin{document}


\title{Decoherence in Quantum Markov Chains}
\author{Raqueline A. M. Santos, Renato Portugal, Marcelo D. Fragoso\\
\\
{\small Laborat\'{o}rio Nacional de Computa\c{c}\~{a}o Cient\'{i}fica (LNCC)} \\
{\small Av. Get\'{u}lio Vargas 333, 25651-075, Petr\'{o}polis, RJ, Brazil}\\
{\small \{raqueline,portugal,frag\}@lncc.br}\\
}

\date{}
\maketitle

\begin{abstract}
It is known that under some assumptions the hitting time in quantum Markov chains is quadratically smaller than the hitting time in classical Markov chains. This work extends this result for decoherent quantum Markov chains. The decoherence is introduced using a percolation-like graph model, which allows us to define a decoherent quantum hitting time and to establish a decoherent-intensity range for which the decoherent quantum hitting time is quadratically smaller than the the classical hitting time. The detection problem under decoherence is also solved with quadratic speedup in this range.
\end{abstract}

\section{Introduction}
In Computer Science, Markov chains are employed in randomized algorithms such as searching algorithms on graphs. The expected time to reach a vertex for the first time, known as {\it hitting time}, plays an important role in those algorithms as the running time to find a solution. For instance, randomized algorithms are used to address the $k$-SAT and the graph connectivity problem~\cite{Motwani:1995}.

In the classical case, Markov chains and random walks are equivalent formalisms. In the quantum case, there are three versions of quantum walks: 1)~discrete-time quantum walks~\cite{Aharonov:1993}, 2)~continuous-time quantum walks~\cite{Farhi:1998}, and 3)~Szegedy's formalism~\cite{Szegedy:2004}. All of them have been used for developing quantum algorithms that outperform their classical versions~\cite{Shenvi:2003,Ambainis:2005,Childs:2004,Ambainis:2004,Krovi:2010,Venegas:2012}. These models use Hilbert spaces of different size and they seem to be non-equivalent. Reference~\cite{Shikano:2013} reviews the discrete-time quantum walk and its relationship with the continuous-time quantum walk. Our focus is the Szegedy's formalism, which is also called quantum Markov chain.

Szegedy~\cite{Szegedy:2004} showed that the quantum hitting time has a quadratic improvement over the classical one to detect a set of marked vertices for ergodic and symmetric Markov chains. This model was developed further by Magniez \textit{et al.}~\cite{Magniez:2009}, who described a quantum algorithm with quadratic speedup for finding a marked vertex on reversible, state-transitive Markov chains restricted to the case of only one marked vertex. Recently, Krovi {\it et al.}~\cite{Krovi:2010} showed that the quadratic speedup for finding a marked vertex also holds for any reversible Markov chain using a new interpolating algorithm. Santos and Portugal~\cite{Santos:2010} worked out the details of Szegedy's model on the complete graph, obtaining analytical results.

When implementing quantum systems, decoherence problems are inevitable. It is impossible to completely isolate a real physical system and interactions with the environment will reduce or possibly destroy quantum coherence. These generally undesired effects are present in quantum walk implementations. Hence, it is crucial to understand how decoherence affects them. Many papers address this issue. Brun {\it et al.}~\cite{Brun:2003} analyzed the transition of a discrete-time quantum walk to a classical random walk by introducing decoherence in the coin operator. Kendon and Tregenna~\cite{Kendon:2003} also studied decoherence in the coin operator and highlighted some useful application of decoherence. Romanelli {\it et al.}~\cite{Romanelli:2005} analyzed the decoherence produced by random broken links of a one-dimensional lattice. This technique was generalized for the two-dimensional case by Oliveira {\it et al.}~\cite{Oliveira:2006}. Alagic and Russel~\cite{Alagic:2005} analyzed the effect of performing independent measurements on the continuous-time quantum walk on the hypercube. A review on decoherence in quantum walks can be found in~\cite{Kendon:2007}.

Decoherence inspired by percolation graphs was analyzed in many papers~\cite{Romanelli:2005,Oliveira:2006,Xu:2008,Leung:2010,Lovett:2011}, using the discrete-time and the continuous-time quantum walk models.  Recently, Kollar {\it et al}~\cite{Kollar:2012} analyzed the asymptotic behavior of discrete-time quantum walks affected by a similar model of decoherence. Previous to our work, decoherence in Szegedy's formalism was studied by Chiang and Gomez~\cite{Chiang:2010}, who analyzed the sensibility to perturbation due to system's precision limitations, by adding a symmetric matrix $E$, representing the noise, to the transition probability matrix. The quadratic speedup vanishes when the magnitude of the noise $||E||$ is greater or equal to $\Omega(\delta(1-\delta\epsilon))$, where $\delta$ is the spectral gap of the transition probability matrix of the graph and $\epsilon$ is the ratio between the number of marked vertices and the number of vertices of the graph. 

Our contribution in this paper is to investigate the effects of decoherence on quantum Markov chains using percolation-like graphs. Our focus is on the quantum hitting time behavior in the presence of this kind of decoherence. By performing averages over all possible evolution operators affected by the decoherence, we are able to define a decoherent quantum hitting time. We obtain a percolation-probability range in which the quadratic speedup on the quantum hitting time for ergodic and symmetric Markov chains is still valid. We show that the problem of detecting a set of marked vertices is also solved with a quadratic speedup in this range.

This paper is organized as follows. In Sec.~\ref{sec:qw}, we review Szegedy's quantum walk and the definition for the quantum hitting time. In Sec.~\ref{sec:dm}, we describe the decoherence model on Szegedy's quantum walk and we define the decoherent quantum hitting time. The bound obtained for the decoherent quantum hitting time is presented in Sec.~\ref{sec:teo}. In Sec.~\ref{sec:detect}, we analyze the Detection Problem and, in Sec.~\ref{sec:conc}, we draw our conclusions.

\section{{Quantum Markov Chains}}\label{sec:qw}

Szegedy~\cite{Szegedy:2004} has proposed a quantum walk driven by reflection operators in an associated bipartite graph obtained from the original one by a process of duplication.
Let $\Gamma(X,E)$ be a connected, undirected and non-bipartite graph, where $X$ is the set of vertices and $E$ is the set of edges. Consider that the stochastic matrix $P$ associated with this graph is symmetric and $p_{xy}$ are its components. Define a bipartite graph associated with $ \Gamma (X, E) $ through a process of duplication. $ X $ and $ Y $ are the sets of vertices of same cardinality of the bipartite graph. Each edge $ \{x_i, x_j\} $ in $ E $ of the original graph $\Gamma(X,E)$ is converted into two edges in the bipartite graph $\{x_i,y_j\}$ and $\{y_i,x_j\}$.

To define a quantum walk in the bipartite graph, we associate with the graph a Hilbert space ${\cal H}^{n^2} = {\cal H}^{n}\otimes {\cal H}^{n} $, where $ n = | X | = | Y | $. The computational basis of the first component is $ \big \{\ket {x}: x \in X \big \} $ and of the second $ \big \{\ket{y}: y \in Y \big \} $. The computational basis of $ {\cal H}^{n^2} $ is $ \big \{\ket {x, y}: x \in X, y \in Y \big \} $.
The quantum walk on the bipartite graph is defined by the evolution operator $ U_P $ given by
\begin{equation}\label{ht_U_ev}
    U_P := \cal R_B \, \cal R_A,
\end{equation}
where
\begin{eqnarray}
  \cal R_A &=& 2AA^{T} - I_{n^2}, \label{ht_RA}\\
  \cal R_B &=& 2BB^{T} - I_{n^2}, \label{ht_RB}
\end{eqnarray}
with the operators $ A: {\cal H}^n \rightarrow {\cal H}^{n^2} $ and $ B: {\cal H}^n \rightarrow {\cal H}^{n^2} $ defined as follows
\begin{eqnarray}
  A &=& \sum_{x\in X} \ket{\Phi_x}\bra{x}, \label{ht_A}\\
  B &=& \sum_{y\in Y} \ket{\Psi_y}\bra{y}, \label{ht_B}
\end{eqnarray}
and
\begin{eqnarray}
  \ket{\Phi_x} &=& \ket{x}\otimes \left(\sum_{y\in Y} \sqrt{p_{x y}} \, \ket{y}\right), \label{ht_alpha_x} \\
  \ket{\Psi_y}  &=&  \left(\sum_{x\in X} \sqrt{p_{y x}} \, \ket{x}\right)\otimes \ket{y}. \label{ht_beta_y}
\end{eqnarray}

In the bipartite graph, an application of $ U_P $ corresponds to two quantum steps of the walk, from $X$ to $Y$ and from $Y$ to $X$. We have to take the partial trace over the space associated with $ Y $ to obtain the state on the set $ X $.

\subsection{Coherent Quantum Hitting Time}\label{sec:qht}

Instead of using the stochastic matrix $P$, Szegedy defined the {\it quantum hitting time} by using a modified evolution operator $U_{P'}$ associated with a modified stochastic matrix $P^\prime$, which is given by
\begin{equation}\label{ht_pprime}
    p_{x y}^\prime = \left\{
                       \begin{array}{ll}
                         p_{x y}, & \hbox{$x\not\in M$;} \\
                         \delta_{x y}, & \hbox{$x\in M$,}
                       \end{array}
                     \right.
\end{equation}
where $M$ is the set of marked vertices. The initial condition of the quantum walk is
\begin{equation}\label{ht_inicond}
    \ket{\psi(0)} = \frac{1}{\sqrt n} \sum_{
\begin{subarray}{c}
{x\in X}\\
y\in Y
\end{subarray}
} \sqrt{p_{xy}} \ket{x,y}.
\end{equation}
Note that $\ket{\psi(0)}$ is an eigenvector of $ U_P $ with eigenvalue $1$. However, $\ket{\psi(0)}$ is not an eigenvector of $ U_{P^\prime} $ in general.

\begin{definition} \cite{Szegedy:2004}  The \textit{quantum hitting time} $H_{P,M}$ of a quantum walk with evolution operator $U_P$ given by Eq.~(\ref{ht_U_ev}) and initial condition $\ket{\psi(0)}$ is defined as the least number of steps $ T $ such that
\begin{equation}
    F(T) \geq 1-\frac{m}{n},
\end{equation}
where $ m $ is the number of marked vertices, $ n $ is the number of vertices of the original graph
and $F(T)$ is
\begin{equation}\label{ht_D_T}
    F(T) = \frac{1}{T+1}\sum_{t=0}^T \Big\| U_{P^\prime}^t\ket{\psi(0)}-\ket{\psi(0)}\Big\|^2,
\end{equation}
where $U_{P^\prime}^t$ is the evolution operator after $t$ steps using the modified stochastic matrix.
\end{definition}

The key operator to find the spectral decomposition of $U_{P^\prime}$ is $C = A^{T}B$~\cite{Szegedy:2004}. The components of $C_{x y}$ are $\sqrt{p_{x y} p_{y x}}$. We have to replace $p_{x y}$ by $p_{x y}^\prime$. Then, we have
\begin{equation}\label{ht_C_Cy}
    C = \left(\begin{array}{cc}
  P_M & 0 \\
  0 & I_m \\
\end{array}\right),
\end{equation}
where $P_M$ is the matrix that we obtain from $P$ by deleting its rows and columns indexed by $M$. The eigenvectors and eigenvalues of the evolution operator are obtained from the singular decomposition of the operator C. Therefore, its singular values and vectors are directly related to the spectral decomposition of $P_M$ and we can write each of these singular values as $\cos\theta$ for some $0\leq \theta \leq \frac{\pi}{2}$.

Lemma~6 of Szegedy's paper~\cite{Szegedy:2004} states that $H_{P,M}$ is at most
\begin{equation}\label{eq:HT}
\frac{100}{1-\frac{m}{n}}\sum_{k=1}^{n-m}\frac{\nu_{k}^{2}}{\sqrt{1-\lambda_{k}'}},
\end{equation}
where $\lambda_{1}',\hdots,\lambda_{n-m}'$ are the eigenvalues of $P_{M}$ and $\ket{v_{1}'},\hdots, \ket{v_{n-m}'}$ are the associated normalized eigenvectors. Coefficients $\nu_{k}$ are defined such that $\ket{\hat{u}} = \sum_{k=1}^{n-m}\nu_{k}\ket{v_{k}'}$ for $\ket{\hat{u}} = \dfrac{1}{\sqrt{n}}\textrm{\bf 1}$, where $\textrm{\bf 1}$ is the $(n-m)$-dimensional vector with entries equal to 1.

\section{Decoherent Quantum Hitting Time}\label{sec:dm}

Percolation model was introduced in the context of a porous media for analyzing liquid flows, and is now a paradigm model of statistical physics~\cite{Stauffer:1994}. The decoherence model inspired by percolation can be explained in the following way: Suppose that a walker is on a vertex of a graph. Before moving to the neighboring vertices, each edge can be removed and each non-edge can be inserted with probability $p$. With probability $1-p$, an edge stays unchanged and the same is valid for a non-edge. After this change in the graph topology, the walker moves following the dynamics of the model. The original graph is reset and the process is repeated over in the next step. This decoherence model is called bond percolation.

The occurrence probability of a given $P_i$ is determined as follows. If $0< p<1$, then $Pr(P_i) = (1-p)^{a_c - a_d}p^{a_d}$, where $a_c = \frac{n(n-1)}{2}$ is the number of edges of the complete graph with $n$ vertices and $a_d$ is the number of edges removed plus the number of edges included to obtain $P_i$ from $P$. If $p=0$, $Pr(P_i=P) = 1$, and $Pr(P_i \neq P) = 0$. And, if $p=1$, we have $Pr(P_i = \bar{P}) = 1$, and $Pr(P_i \neq \bar{P}) = 0$, where $\bar{P}$ is the complement of $P$. The evolution under maximum decoherence occurs when $p=1/2$, since at each step  a random graph is selected.

Another decoherence model can be analyzed at this point. If only removal of edges are allowed (no insertions), $a_c$ must be replaced by the number of edges of the original graph and $a_d$ will be the number of removed edges.

The dynamics under decoherence has a new behavior, because at each step the graph changes. This process changes the transition probability matrix, which also changes the evolution operator. Therefore, instead of having a usual walk evolving as $\ket{\psi(t)} = U_{P}^t\ket{\psi(0)}$, we have
\begin{equation}
\ket{\psi(t)} = U_{P_t} U_{P_{t-1}}\cdots U_{P_1}\ket{\psi(0)} =: U_{\vec{P}_t}\ket{\psi(0)},
\end{equation}
where $\vec{P}_t = (P_1,\hdots,P_{t-1},P_{t})$ and $U_{\vec{P}_t} = U_{P_t} U_{P_{t-1}}\cdots U_{P_1}$. $P_i$'s are not necessarily equal and they are independent and identically distributed. They are obtained from $P$ and for each $P_i$ we have a $P'_i$ associated depending on the cardinality of $M$.

In this context, it is useful to define an operator that will represent the behavior of the operators affected by the decoherence. Let
\begin{equation}
\bar{U}_{dec} := \sum_{P}Pr(P)U_P,
\end{equation}
be the operator obtained by performing an average over all possible evolution operators affected by the decoherence. The following result shows that the average over all possible sequences $\vec{P}$, with size $T$, according to its probability distribution, is equal to $\bar{U}_{dec}^T$.

\begin{lemma}\label{lemma:UbarT} For $t\leq T$ we have
\begin{equation}
\sum_{\vec{P}_T}Pr(\vec{P}_T)U_{\vec{P}_t}= \bar{U}_{dec}^t.
\end{equation}
\end{lemma}

\begin{proof}
Since $ Pr(\vec{P}_T) = \prod_{i=1}^TPr(P_i)$, we have
\begin{eqnarray*}
&&\sum_{\vec{P}_T}Pr(\vec{P}_T)U_{P_t}U_{P_{t-1}}\cdots U_{P_1} = \sum_{\vec{P}_T} \prod_{i=1}^TPr(P_i)U_{P_t}U_{P_{t-1}}\cdots U_{P_1}\\
&&=\sum_{P_T}\sum_{P_{T-1}}\cdots\sum_{P_2}\left( \prod_{i=2}^TPr(P_i)\right)U_{P_t}U_{P_{t-1}}\cdots U_{P_2}\left(\sum_{P_1}Pr(P_1)U_{P_1}\right)\\
&&=\sum_{P_T}\sum_{P_{T-1}}\cdots\sum_{P_{t+1}} Pr(P_T)Pr(P_{T-1})\cdots Pr(P_{t+1})\bar{U}_{dec}^t\\
&&= \bar{U}_{dec}^t
\end{eqnarray*} \qed
\end{proof}

To define the quantum hitting time for the decoherent evolution, we have to do an average over all possible sequences $\vec{P}$. Define,
\begin{equation}\label{eq:ofdec}
F_{dec}(T) := \sum_{\vec{P}_T}Pr(\vec{P}_T)\left( \frac{1}{T+1}\sum_{t=0}^T\Big\|U_{\vec{P}_t}\ket{\psi(0)}-\ket{\psi(0)} \Big\|^2\right).
\end{equation}

\begin{lemma}\label{lemma:fdec}
\begin{equation}\label{eq:F_dec}
F_{dec}(T) =2-\frac{2}{T+1}\sum_{t=0}^T\bracket{\psi(0)}{\bar{U}_{dec}^t}{\psi(0)}.
\end{equation}
\end{lemma}

\begin{proof}
 Expanding Eq.~\eqref{eq:ofdec} and using that the initial condition and the evolution operators are real, we obtain
\begin{eqnarray}
F_{dec}(T) &=&\sum_{\vec{P}_T}Pr(\vec{P}_T)\left(\frac{1}{T+1}\sum_{t=0}^T\left(2-2\bracket{\psi(0)}{U_{\vec{P}_t}}{\psi(0)} \right)\right)\nonumber\\
&=&\frac{1}{T+1}\sum_{t=0}^T\left(2-2\bra{\psi(0)}{\left(\sum_{\vec{P}_T}Pr(\vec{P}_T)U_{\vec{P}_t}\right)}\ket{\psi(0)} \right).\label{eq:Fdecfromlemma}
\end{eqnarray}
Using Lemma~1, we obtain Eq.~\eqref{eq:F_dec}.
\qed\end{proof}

Now, we can naturally define the decoherent quantum hitting time (DQHT), using the expression of $F_{dec}$ obtained in Lemma~2.

\begin{definition}
The \textit{decoherent quantum hitting time} $H_{P,M}^{dec}$ of a quantum walk with evolution operator $U_P$ given by Eq.~(\ref{ht_U_ev}) and initial condition $\ket{\psi(0)}$ given by Eq.~(\ref{ht_inicond}) is defined as the least number of steps $ T $ such that
\begin{equation}
    F_{dec}(T) \geq 1-\frac{m}{n}.
\end{equation}
\end{definition}
Note that when $p=0$, we have the original definition, since $\bar{U}_{dec} = U_{P'}$.

\section{Bounds on DQHT}\label{sec:teo}

It is important to mention that we are considering ergodic Markov chains with symmetric transition matrix. Thus, the following result generalizes the result from Szegedy~\cite{Szegedy:2004} by introducing a decoherence term.

\begin{theorem}\label{teo:dqht}
The \textit{decoherent quantum hitting time} $H_{P,M}^{dec}$ of a quantum walk with evolution operator $U_P$, given by Eq.~(\ref{ht_U_ev}), initial condition $\ket{\psi(0)}$, and $p \leq \frac{1}{400a_cE}$ where
\begin{equation}\label{eq:E}
    E = \frac{1}{1-\frac{m}{n}}\sum_{k=1}^{n-m} \frac{\nu_k^2}{\arccos(\lambda_k')},
\end{equation}
is at most
\begin{equation}\label{eq:decHT}
\frac{8}{1-\frac{m}{n}}\sum_{k=1}^{n-m}\frac{\nu_{k}^{2}}{\sqrt{1-\lambda_{k}'}}+
\frac{1434\,a_c\,p}{\big(1-\frac{m}{n}\big)^2}\left(\sum_{k=1}^{n-m}\frac{\nu_{k}^{2}}{\sqrt{1-\lambda_{k}'}}\right)^2.
\end{equation}
\end{theorem}

\begin{proof}
Using expression \eqref{eq:decHT} and replacing $\lambda_{k}'$ by $\cos\theta_k$, define
\begin{equation}
T(p)= 8\sum_{k=1}^{n-m}\frac{\nu_{k}^{2}}{(1-\epsilon)\sqrt{{1-\cos\theta_{k}}}}+1434 a_cp\left(\sum_{k=1}^{n-m}\frac{\nu_{k}^{2}}{(1-\epsilon)\sqrt{{1-\cos\theta_k}}}\right)^2,
\end{equation}
where $\epsilon = \frac{m}{n}$.  For now on, we are going to omit the dependence of $p$ for the time $T$. Using $1-\cos\alpha \geq 2\alpha^2/5$ ($\alpha \in (0,\pi/2]$), we obtain
\begin{equation}
T \leq 13\sum_{k=1}^{n-m}\dfrac{\nu_{k}^{2}}{(1-\epsilon)\theta_k} + 3585 a_cp\left(\sum_{k=1}^{n-m}\dfrac{\nu_{k}^{2}}{(1-\epsilon)\theta_k}\right)^2.
\end{equation}
Using $1-\cos\alpha \leq \alpha^2$, we obtain
\begin{equation}
T \geq 8\sum_{k=1}^{n-m}\dfrac{\nu_{k}^{2}}{(1-\epsilon)\theta_k} + 1434 a_cp\left(\sum_{k=1}^{n-m}\dfrac{\nu_{k}^{2}}{(1-\epsilon)\theta_k}\right)^2.
\end{equation}
Using $E=\sum_{k=1}^{n-m}\frac{\nu_k^2}{(1-\epsilon)\theta_k}$, we obtain
\begin{equation}\label{eq:bT}
8E+1434a_cpE^2\leq T\leq 13E+3585a_cpE^2.
\end{equation}

Let us write the initial condition as $\ket{\psi(0)} = \ket{\psi_{M^\bot}} + \ket{\psi_{M}}$, where
 \begin{align}
\ket{\psi_{M^\bot}} &= \dfrac{1}{\sqrt{n}}\sum_{\sumid{x \in X\backslash M}{y\in X}}\sqrt{p_{xy}}\ket{x}\ket{y},\\
\ket{\psi_{M}} &= \dfrac{1}{\sqrt{n}}\sum_{\sumid{x \in M}{y\in X}}\sqrt{p_{xy}}\ket{x}\ket{y}. \label{eq:psiM}
\end{align}
Note that $\|\ket{\psi_{M^\bot}}\|^2 = 1-\epsilon$ and $\|\ket{\psi_{M}}\|^2 = \epsilon$.
We want to show that $F_{dec}(T)$ is greater than or equal to $1-\frac{m}{n}$ when $T$ is in the range \eqref{eq:bT}. Using Eq.~(\ref{eq:Fdecfromlemma}) and $\ket{\psi(0)} = \ket{\psi_{M^\bot}} + \ket{\psi_{M}}$, we can write $F_{dec}(T)$ as
\begin{equation}
F_{dec}(T) = 2-2(G_M+G_{M,M^\bot}+G_{M^\bot}),
\end{equation}
where
\begin{eqnarray}
&&G_M = \frac{1}{T+1}\SumPT\Sumt\bracket{\psi_{M}}{\UPt}{\psi_{M}},\label{gm}\\
&&G_{M,M^\bot} = \frac{1}{T+1}\SumPT\Sumt\left(\bracket{\psi_{M^\bot}}{\UPt}{\psi_{M}}+\bracket{\psi_{M}}{\UPt}{\psi_{M^\bot}}\right),\label{gmmbot}\\
&&G_{M^\bot} = \frac{1}{T+1}\SumPT\Sumt\bracket{\psi_{M^\bot}}{\UPt}{\psi_{M^\bot}}.\label{gmbot}
\end{eqnarray}
Let us establish bounds for $G_M$, $G_{M,M^\bot}$ and $G_{M^\bot}$:
\begin{equation}\label{eq:gm2}
G_M \leq \frac{1}{T+1}\SumPT\Sumt\braket{\psi_{M}}{\psi_{M}} = \epsilon.
\end{equation}
Expression~(\ref{gmmbot}) can be expanded in two terms
\begin{equation}\label{eq:gmm}
\begin{split}
G_{M,M^\bot}&=\frac{1}{T+1}Pr\left(\vec{P}_T=(P',...,P')\right)\Sumt\Big( \bracket{\psi_{M^\bot}}{U_{P'}^t}{\psi_{M}}+\bracket{\psi_{M}}{U_{P'}^t}{\psi_{M^\bot}}\Big)\\
&+\frac{1}{T+1}\sum_{\vec{P}_T\neq (P',...,P')}Pr(\vec{P}_T)\Sumt\left(\bracket{\psi_{M^\bot}}{\UPt}{\psi_{M}}+\bracket{\psi_{M}}{\UPt}{\psi_{M^\bot}}\right).
\end{split}
\end{equation}
The first term of Eq.~(\ref{eq:gmm}) is zero because $\ket{\psi_{M^\bot}}$ is in the space spanned by $A\ket{w_k}$ and $B\ket{v_k}$, that is invariant under the action of $U_{P'}$. Taking $\epsilon \leq 1/2$,
\begin{equation*}
\begin{split}
\bracket{\psi_{M^\bot}}{\UPt}{\psi_{M}} +\bracket{\psi_{M}}{\UPt}{\psi_{M^\bot}} &\leq 2\max\left\{\braket{\psi_M}{\psi_M},\braket{\psi_{M^\bot}}{\psi_{M^\bot}} \right\}\\
&= 2\max\{\epsilon,1-\epsilon\} \\
&= 2(1-\epsilon).
\end{split}
\end{equation*}
 Then, using that $(1-p)^{a_cT} = 1-a_cpT+\frac{a_cTp^2}{2}(a_cT-1)+O(p^3)$, we have
\begin{equation}\label{eq:gmmbot2}
G_{M,M^\bot} \leq 2(1-\epsilon)(1-(1-p)^{a_cT})\leq 2(1-\epsilon)a_cpT.
\end{equation}
Finally for $G_{M^\bot}$, we have
\begin{eqnarray*}
G_{M^\bot}&=&\frac{1}{T+1}Pr(\vec{P}_T=(P',...,P'))\Sumt\bracket{\psi_{M^\bot}}{U_{P'}^t}{\psi_{M^\bot}} +\\
&&\frac{1}{T+1}\sum_{\vec{P}_T\neq (P',...,P')}Pr(\vec{P}_T)\Sumt\bracket{\psi_{M^\bot}}{\UPt}{\psi_{M^\bot}} \\
&\leq& \frac{(1-\epsilon)(1-p)^{a_cT}}{T+1}\Sumk\frac{\nu_k^2}{1-\epsilon}\Sumt\cos(2t\theta_k)+ (1-\epsilon)(1-(1-p)^{a_cT}).
\end{eqnarray*}
From Eq.~(13) of Szegedy's paper~\cite{Szegedy:2004}, we know that
\begin{equation}
\frac{1}{T+1}\Sumk\frac{\nu_k^2}{1-\epsilon}\Sumt\cos(2t\theta_k) \leq \frac{1}{T+1}\Sumk\frac{4\nu_k^2}{(1-\epsilon)\theta_k}.
\end{equation}
 Using again the Taylor series expansion for $(1-p)^{a_cT}$ and that $\frac{4E}{T}\leq \frac{1}{2}$, we have
\begin{equation}\label{eq:gmbot2}
G_{M^\bot}\leq(1-\epsilon)\left(\left(1-a_cpT\right)\frac{4E}{T}+a_cpT \right).
\end{equation}

From Eqs.~(\ref{eq:gm2}), (\ref{eq:gmmbot2}), and (\ref{eq:gmbot2}), we obtain
\begin{equation}
G_M+G_{M,M^\bot}+G_{M^\bot}
\leq \epsilon + (1-\epsilon)\left(\frac{4E}{T}-4a_cpE + 3a_cpT  \right).
\end{equation}
 Using that $T$ is in the range~(\ref{eq:bT}),
\begin{equation}
\frac{4E}{T}-4a_cpE + 3a_cpT  \leq \frac{1}{2(1+179.25a_cpE)}+35a_cpE+10755a_c^2p^2E^2 \leq \frac{1}{2},
\end{equation}
if we choose $p \leq \frac{1}{400a_cE}$. Then, we have
\begin{equation}
G_M+G_{M,M^\bot}+G_{M^\bot} \leq \frac{\epsilon+1}{2},
\end{equation}
and
\begin{equation}
F_{dec}(T) \geq 2 - 2\left(\frac{\epsilon+1}{2}\right) = 1-\epsilon.
\end{equation}
\qed\end{proof}

\begin{corollary}\label{col:dqht}
The decoherent quantum hitting time $H_{P,M}^{dec}$ of $U_P$ with respect to any $M \subseteq X$ with $m \leq n/2$ and $0\leq p \leq \frac{1}{400a_cE}$, where $$E = \frac{1}{1-\frac{m}{n}}\sum_{k=1}^{n-m} \frac{\nu_k^2}{\arccos(\lambda_k')},$$ is in
$O\left(\frac{1}{\sqrt{1-\lambda(P_M)}}\right)$, where $\lambda(P_M)$ is the largest eigenvalue of $P_M$.
\end{corollary}
\begin{proof}
Using $1-\cos\alpha \geq 2\alpha^2/5$, we obtain

\begin{equation}
E \geq \frac{1}{2}\sum_{k=1}^{n-m}\frac{\nu_k^2}{1-\epsilon}\sqrt{\frac{1}{1-\cos\theta_k}}
\end{equation}
and
\begin{equation}\label{eq:p}
p \leq \frac{1}{200a_c\sum_{k=1}^{n-m}\frac{\nu_k^2}{1-\epsilon}\sqrt{\frac{1}{1-\lambda_k'}}}.
\end{equation}
By substituting Eq.~(\ref{eq:p}) for the expression of $H_{P,M}^{dec}$ given by (\ref{eq:decHT}) and since
\begin{equation}
\sum_k^{n-m} \nu_k^2 \sqrt{\frac{1}{1-\lambda_k'}} \leq \sqrt{\sum_k^{n-m} \frac{\nu_k^2}{1-\lambda_k'}} \leq \sqrt{\frac{1}{1-\lambda(P_M)}},
\end{equation}
we conclude that $H_{P,M}^{dec} $ is in $O\left(\frac{1}{\sqrt{1-\lambda(P_M)}}\right)$.
\qed\end{proof}

Expression~(\ref{eq:decHT}) of Theorem~1 shows that the decoherent hitting time has an additional term that is proportional to the square of the usual term. If $p$ is small enough, the contribution of the new term for the hitting time scales as a linear function in terms of $p$. Corollary~1 describes a range of $p$ such that the quadratic speedup is valid.

\section{The Detection Problem}\label{sec:detect}
Let ${\cal M} \subseteq 2^X$ be a set of non-empty subsets of $M$. The Detection Problem stands as the problem of finding out if the set of marked vertices is either empty or belongs to $\cal M$.

\begin{theorem}\label{teo:detect}
Assume that $T$ is an upper bound for
\begin{equation}\label{eq:decHT2}
16\sum_{k=1}^{n-|M|}{\dfrac{\nu_{k}^{2}}{\sqrt{1-\lambda_{k}'}}}+5736 a_cp\left(\sum_{k=1}^{n-|M|}{\dfrac{\nu_{k}^{2}}{\sqrt{1-\lambda_{k}'}}}\right)^2,
\end{equation}
where $M$ runs through all elements of $\cal M$ ($\lambda'_k$ and $\nu_k$ depend on $M$) and $p$ obeys the same inequality of Theorem~1. Then, the Detection Problem can be solved within time $T$ with bounded two-sided error.
\end{theorem}

Select $0\leq t\leq T$ uniformly random and let $U_1,\cdots, U_t$ be unitary operators of the dynamics under decoherence. Algorithm~\ref{alg:detect} creates the state
\begin{equation}\label{eq:alg}
\frac{1}{2}\ket{0}\left(\ket{\psi(0)} + U_{t}U_{t-1}...U_{1}\ket{\psi(0)}\right)+\frac{1}{2}\ket{1}\left(\ket{\psi(0)} - U_{t}U_{t-1}...U_{1}\ket{\psi(0)}\right),
\end{equation}
which has an additional control register of one qubit.
\begin{algorithm}[!htb]
\SetKwInOut{Input}{input}\SetKwInOut{Output}{output}
\Input{$U_1,\cdots, U_t$}
\Output{0 ($M = \emptyset$) or 1 ($M \in \cal M$)}
\Begin{
prepare the state: $\ket{0}\ket{\psi(0)}$\;
apply $H$ to the control register\;
apply $C(U_1),\cdots,C(U_t)$, where $C(U)$ is the controlled-U\;
apply $H$ to the control register\;
measure in computational basis\;\label{step5}
\If{\emph{the control register is 1}}
{\Return{1}}
\Else{\Return{0}}
}
\caption{Detect if marked}\label{alg:detect}
\end{algorithm}
Note that, when using Algorithm~\ref{alg:detect}, we do not know if $U_P$ or $U_{P'}$ are being used.
By measuring the state given by Eq.~(\ref{eq:alg}) in the computational basis, the probabilities of having the control register in state 0 and 1 are
\begin{eqnarray}
P^{(0)} &=& \frac{1}{4}\Big|\Big|\ket{\psi(0)} + U_{t}U_{t-1}...U_{1}\ket{\psi(0)}\Big|\Big|^2,\label{eq:p0}\\
P^{(1)} &=& \frac{1}{4}\Big|\Big|\ket{\psi(0)} - U_{t}U_{t-1}...U_{1}\ket{\psi(0)}\Big|\Big|^2.\label{eq:p1}
\end{eqnarray}

Now, analyzing the algorithm after measuring (Step~\ref{step5}), if $M = \emptyset$, the control register is in state $\ket{0}$ with probability at least $(1-p)^{a_cT}$, that is the probability of having $U = U_P$.
If $M \in {\cal M}$, the control register is in state $\ket{1}$ with probability at least
\begin{equation}\label{eq:p12}
\frac{1}{4(T+1)}\SumPT\Sumt\Big|\Big|\ket{\psi(0)} - \UPt\ket{\psi(0)}\Big|\Big|^2 \geq \frac{1}{4}\left(1-\frac{m}{n}\right).
\end{equation}
Eq.~(\ref{eq:p12}) is obtained from Eq.~(\ref{eq:p1}) by performing two averages: an average on time, because the algorithm chooses a time $t$ at random, and another average on the possible sequences $\vec{P}_T$ since the algorithm can be affected by the decoherence. Therefore, using that $\frac{m}{n}\leq \frac{1}{2}$, we obtain 1 in the control register with probability at least $\frac{1}{8}$, which means that there is at least one marked element.

This result can be improved if we consider the decoherence model which allows only removal of edges from the graph (insertions are not allowed). We can solve the Detection Problem within time $T$ with bounded one-sided error, because the initial condition will be invariant under the action of any $U_{P_i}$ when $M$ is empty. In this case, the probability to obtain 0 when there are no marked vertices is 1.

\section{Conclusions}\label{sec:conc}
We have proposed a decoherence model on Szegedy's quantum walk inspired by percolation graphs. This model is characterized by the possibility of removing or inserting edges at each time step with probability $p$. The graph probability matrix and the evolution operator change at each time step. We were able to define a decoherent quantum hitting time by using a new operator, which is obtained by performing an average over all possible evolution operators affected by the decoherence. Note that when the percolation probability $p$ is zero, the evolution operator is equal to Szegedy's original definition and the decoherent hitting time is equal to the original quantum hitting time~\cite{Szegedy:2004}.

We have proved that, for $p$ small enough, the decoherent quantum hitting time has an additional term which depends linearly on $p$, preserving the quadratic speedup on the quantum hitting time over the classical case for some range of $p$.
Moreover, the Detection Problem can be solved in time of the order of the DQHT with bounded two-sided error, and with bounded one-sided error if the decoherence model does not allow edge insertions.

It is interesting to analyze the behavior of the hitting time in the complement of the original graph. This case occurs when $p = 1$ and $\Udec = U_{\bar{P}'}$. Since the initial condition is associated to the original graph $P$, it is easy to show that $\ket{\psi(0)}$ is invariant under the action of $U_{\bar{P}'}$. This occurs because the initial condition is a superposition over all edges of the original graph and those edges do not exist in its complement. Thus, $\ket{\psi(0)}$ is an eigenvector  of $U_{\bar{P}'}$ with eigenvalue 1
and the quantum hitting time goes to infinity when $p$ goes to 1.

In future works, it is interesting to obtain a better upper bound for the DQHT. The approximations used in Eqs.~(\ref{eq:gm2}), (\ref{eq:gmmbot2}), and (\ref{eq:gmbot2}) help the analytical calculations, but seem to weaken the final value of the bound. It is also interesting to analyze the decoherence effects on spatial search algorithms.

\section*{Acknowledgements}
We thank F.L. Marquezino for fruitful discussions. The authors acknowledge financial support from FAPERJ n. E-26/100.484/2012, CNPq, and CAPES.



\end{document}